\documentclass{article}

\usepackage{authblk}
\usepackage{hyperref}
\usepackage{amssymb}
\usepackage{amsmath}

\usepackage{lineno}
\usepackage{eurosym}
\usepackage[ruled,vlined]{algorithm2e}
\usepackage{subfig}
\usepackage{url}
\usepackage{xcolor}
\usepackage[inline,shortlabels]{enumitem}

\usepackage{graphicx}%
\usepackage{multirow}%
\usepackage{amsmath,amssymb,amsfonts}%
\usepackage{amsthm}%
\usepackage{mathrsfs}%
\usepackage[title]{appendix}%
\usepackage{xcolor}%
\usepackage{textcomp}%
\usepackage{manyfoot}%
\usepackage{booktabs}%
\usepackage{listings}%
\usepackage[round-mode=places,round-precision=2,round-pad=false,group-separator={,},output-decimal-marker={.}]{siunitx}

\newcommand{\F}{\mathcal{F}}

\newcommand{\mathsc}[1]{{\normalfont\textsc{#1}}}

\newcommand{\sky}{\mathsc{Sky}}
\newcommand{\nd}{\mathsc{ND}}
\newcommand{\po}{\mathsc{PO}}

\newcommand{\TOP}{\mathsc{Top}}

\newcommand{\rank}{\mathsf{rank}}

\newcommand{\dominates}{\prec}

\newcommand{\fdominates}{\dominates_{\F}}

\newcommand{\dimensions}{d}
\newcommand{\monotoneFunctions}{\mathtt{MF}}

\newcommand{\positivereals}{\mathbb{R^+}}

\newcommand{\size}{N}

\newcommand{\uniform}{\texttt{UNI}\xspace}
\newcommand{\anticorrelated}{\texttt{ANT}\xspace}
\newcommand{\nba}{\texttt{NBA}\xspace}

\newcommand{\step}{\mu}

\newcommand{\mycomment}[1]{}

\newtheorem{theorem}{Theorem}

\newtheorem{definition}{Definition}
\newtheorem{example}[theorem]{Example}

\newcommand{\logSep}{\,.\,\,}

\def\codeif{\mbox{\upshape\textbf{if}}}

\def\codefor{\mbox{\upshape\textbf{for}}}

\def\codein{\mbox{\upshape\textbf{in}}}
\def\codewhile{\mbox{\upshape\textbf{while}}}
\def\codetrue{\mbox{\upshape\textbf{true}}}

\def\codereturn{\mbox{\upshape\textbf{return}}}
\def\codebreak{\mbox{\upshape\textbf{break}}}
\def\codecontinue{\mbox{\upshape\textbf{continue}}}

\providecommand{\keywords}[1]{\textbf{\textit{Keywords---}} #1}

\begin{document}

\title{Computing the Non-Dominated Flexible Skyline in Vertically Distributed Datasets with No Random Access}

\author{Davide Martinenghi}

\affil{Politecnico di Milano, DEIB\\Piazza Leonardo 32, 20133 Milan, Italy.\\email:
\url{davide.martinenghi@polimi.it}
}

\date{}

\maketitle

In today's data-driven world, algorithms operating with vertically distributed datasets are crucial due to the increasing prevalence of large-scale, decentralized data storage. These algorithms enhance data privacy by processing data locally, reducing the need for data transfer and minimizing exposure to breaches. They also improve scalability, as they can handle vast amounts of data spread across multiple locations without requiring centralized access.
Top-$k$ queries have been studied extensively under this lens, and are particularly suitable in applications involving healthcare, finance, and IoT, where data is often sensitive and distributed across various sources.
Classical top-$k$ algorithms are based on the availability of two kinds of access to sources: \emph{sorted access}, i.e., a sequential scan in the internal sort order, one tuple at a time, of the dataset; \emph{random access}, which provides all the information available at a data source for a tuple whose id is known.
However, in scenarios where data retrieval costs are high or data is streamed in real-time or, simply, data are from external sources that only offer sorted access, random access may become impractical or impossible, due to latency issues or data access constraints.
Fortunately, a long tradition of algorithms designed for the ``no random access'' (NRA) scenario exists for classical top-$k$ queries. Yet, these do not cover the recent advances in ranking queries, proposing hybridizations of top-$k$ queries (which are preference-aware and control the output size) and skyline queries (which are preference-agnostic and have uncontrolled output size).
The \emph{non-dominated flexible skyline} ($\nd$) is one such proposal, trying to get the best of top-$k$ and skyline queries. We introduce an algorithm for computing $\nd$ in the NRA scenario, prove its correctness and optimality within its class, and provide an experimental evaluation covering a wide range of cases, with both synthetic and real datasets.

\section{Introduction}
\label{intro}

In today's data-centric world, algorithms designed for vertically distributed datasets are vital due to the growing trend of large-scale, decentralized data storage. These algorithms enhance data privacy by processing information locally, thereby reducing the need for data transfers and minimizing the risk of breaches. They also improve scalability by efficiently managing large volumes of data distributed across various locations without relying on centralized access.
Top-$k$ queries (the fundamental tool to tackle multi-objective optimization by transforming the problem into a single-objective problem through a \emph{scoring function}), have been studied extensively under this lens, and are particularly suitable in applications involving healthcare, finance, and IoT, where data is often sensitive and distributed across various sources.
Classical top-$k$ algorithms, such as~\cite{DBLP:conf/pods/Fagin96}, are based on the availability of two kinds of access to sources: \emph{sorted access}, i.e., a sequential scan in the internal sort order, one tuple at a time, of the dataset; \emph{random access}, which provides all the information available at a data source for a tuple whose id is known.
However, in scenarios where data retrieval costs are high or data is streamed in real-time or, simply, data are from external sources that only offer sorted access, random access may become impractical or impossible, due to latency issues or data access constraints.

Fortunately, a long tradition of algorithms designed for the ``no random access'' (NRA) scenario exists for classical top-$k$ queries~\cite{DBLP:conf/pods/FaginLN01}. Yet, these do not cover the recent advances that have proposed hybridizations of top-$k$ queries and skyline queries~\cite{DBLP:conf/icde/BorzsonyiKS01} (the other common tool for multi-criteria analysis, based on the notion of \emph{dominance}), which try to get the best of both worlds.
Indeed, although both share the overall goal to identify the best objects in a dataset, the way they work is totally different: top-$k$ queries are preference-aware and control the output size, while skylines are preference-agnostic and have uncontrolled output size. \emph{Flexible skylines}~\cite{DBLP:journals/pvldb/CiacciaM17,DBLP:journals/tods/CiacciaM20} are a popular attempt at reconciling top-$k$ queries and skylines under a unifying perspective, and, in particular, the non-dominated (ND) flexible skyline is a preference-aware generalization of skylines that admits efficient implementations.
This paper refers to a further generalization, which introduces the $k$ parameter into skylines (called $k$-skybands~\cite{DBLP:journals/tods/PapadiasTFS05} and consisting of all tuples dominated by less than $k$ tuples), leading to the notion of non-$k$-dominated flexible skyline ($\nd_k$).

\begin{figure}
\centering
\hspace*{-0.20\textwidth}
\subfloat[][{Dataset in tabular form.\label{fig:example-table}}]
{ 
\scalebox{0.8}{\vbox to 0.45\textwidth{\vfil
   		\begin{tabular}{r|c|p{1cm}r|c|}
\multicolumn{2}{r}{$r_1$}&	&	\multicolumn{2}{r}{$r_2$}\\
\cline{2-2} \cline{5-5}
\texttt{\textit{a}}	&         3.0		& 	&	\texttt{\textit{i}}	&	    1.0\\
\texttt{\textit{d}}	&         4.0		& 	&	\texttt{\textit{g}}	&	    1.5\\
\texttt{\textit{h}}	&         5.0		& 	&	\texttt{\textit{e}}	&	    2.0\\
\texttt{\textit{f}}	&         6.0		& 	&	\texttt{\textit{c}}	&	    3.0\\
\texttt{\textit{e}}	&         6.0		& 	&	\texttt{\textit{b}}	&	    6.0\\
\texttt{\textit{c}}	&         7.0		& 	&	\texttt{\textit{h}}	&	    7.0\\
\texttt{\textit{b}}	&         8.0		& 	&	\texttt{\textit{a}}	&	    8.0\\
\texttt{\textit{i}}	&         8.0		& 	&	\texttt{\textit{d}}	&	    9.0\\
\texttt{\textit{g}}	&         9.0		& 	&	\texttt{\textit{f}}	&	    9.0\\
\cline{2-2} \cline{5-5}
   		\end{tabular}
   		\vfil
 		}
 		}
 		}%
\hspace*{-0.20\textwidth}
\subfloat[][{2D depiction of the locations.}]
{\includegraphics[width=0.49\textwidth]{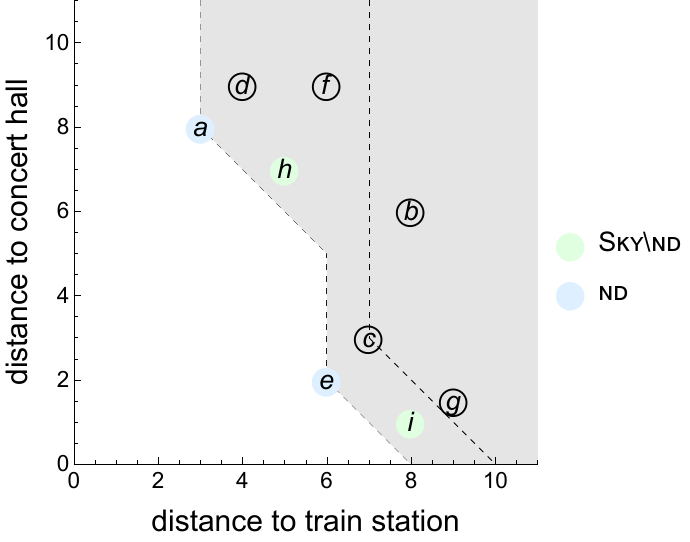}\label{fig:example}}%
 	\caption{A set of locations ranked by distance to given points of interest.}
 	\label{fig:sorted-lists}
 \end{figure}

\begin{example}
To exemplify the notion of flexible skyline, consider the dataset shown in Figure~\ref{fig:example-table}, showing two ranked lists of locations sorted by their distance from a given point of interest (say, train station and concert hall).
The only options that are not dominated by other options in this dataset are \texttt{a}, \texttt{h}, \texttt{e}, \texttt{i}, because no alternative is better than any of them on both distances; this is the skyline of the dataset.
In particular, each of these options is top for at least one monotone scoring function; for instance, \texttt{e} is the top option by average distance ($4$) to the points of interest, i.e., using function $f(x,y)=0.5 x+ 0.5 y$, where $x$ and $y$ are, respectively, the numeric attributes from $r_1$ and $r_2$.

Suppose now that we also consider a preference stating that, while we still want to linearly combine the distances, the distance to the train station is more important, i.e., we refer to the set $\F$ of scoring functions of the form $f(x,y)=w_1 x+ w_2 y$, where $w_1>w_2$.
Under this assumption, \texttt{i} and \texttt{h} are no longer ideal, as no scoring function of the stated form will ever rank them as the top solution; the options \texttt{a} and \texttt{e} are the non-dominated flexible skyline $\nd$ with respect to set $\F$. The gray area shown in Figure~\ref{fig:example} represents the region of points that can never become top-$1$, in this dataset, for any of the functions in $\F$; in particular, \texttt{i} can never beat \texttt{e} and \texttt{h} can never beat \texttt{a} with these functions.

The 2-skyband of the dataset includes all the tuples that are dominated by less than 2 tuples (i.e., 0 or 1) and consists of \texttt{d}, \texttt{c} and \texttt{g}, in addition to all the tuples in the skyline (\texttt{a}, \texttt{h}, \texttt{e}, \texttt{i}). This set contains all tuples that may be in the result set of any top-2 query using a monotone scoring function. For instance, with $f(x,y)=0.5 x+ 0.5 y$, the top-$2$ results are \texttt{e} and \texttt{i}.
However, if we consider the set $\F$, \texttt{g} is now, in a sense, also dominated by \texttt{c}, because no scoring function in $\F$ may rank it better than \texttt{c}, as shown with the dashed line traversing \texttt{c}.
The tuples in the $2$-skyband except for $\texttt{g}$ are the so-called non-$k$-dominated flexible skyline $\nd_k$ with respect to $\F$.
\end{example}

In this paper, after reconsidering the standard NRA approaches, we introduce a novel algorithm for computing the ND flexible skyline in the NRA scenario, prove its correctness and instance optimality (a very strong form of optimality in an I/O sense), and provide an experimental evaluation covering a wide range of datasets, both synthetic and real.

\section{Preliminaries}
\label{sec:prelim}

We refer to datasets with numeric attributes (namely, and without loss of generality, the non-negative real numbers $\positivereals$). A schema $S$ is a set of attributes $\{A_1,\ldots, A_\dimensions\}$, and a tuple $t=\langle v_1,\ldots,v_\dimensions\rangle$ over $S$ is a function associating each attribute $A_i\in S$ with a value $v_i$, also denoted $t[A_i]$, in $\positivereals$; a relation over $S$ is a set of tuples over $S$.

A \emph{scoring function} $f$ over $S=\{A_1,\ldots, A_\dimensions\}$ applies to the attribute values of a tuple $t$ over $S$ and associates $t$ with a numeric score $f(t[A_1],\ldots, t[A_\dimensions])\in \positivereals$, also indicated $f(t)$; $f$ is \emph{monotone} if, for all tuples $t$ and $s$ over $S$, we have
$(\forall A\in S \logSep t[A]\leq s[A])\rightarrow f(t) \leq f(s)$.

We shall consider attributes such as ``cost'', and thus prefer smaller values over higher ones; this is of course an arbitrary convention that could be changed and does not affect generality. Similarly, we also set our preference for lower scores (again, without loss of generality).

The \emph{rank} $\rank(t;r,f)$ of a tuple $t\in r$ according to scoring function $f$ is defined as 
$\rank(t;r,f) = 1 + |\{s \in r \mid f(s) < f(t) \}|$.
\begin{definition}\label{def:top-k}
A \emph{top-$k$ query} takes 
\emph{i)} a relation $r$ over a schema $S$, 
\emph{ii)} a scoring function $f$ over $S$ that totally orders $r$,
\emph{iii)} a positive integer $k$,
and returns the set $\TOP_k(r;f)$ of top-$k$ ranked tuples in $r$ according to $f$, i.e., $\TOP_k(r;f) = \{t \mid t\in r \land \rank(t;r,f)\leq k\}$.
\end{definition}
We assumed that $f$ orders $r$ totally so that no ties occur and the set returned by a top-$k$ query is univocally defined.
When $k=1$, the definition can be expanded as follows:
\begin{equation}\label{eq:top-one}
\TOP_1(r;f) = \{t \mid t\in r \land \forall s\in r\logSep s\neq t \implies f(s) > f(t)\}.
\end{equation}

A \emph{skyline} query~\cite{DBLP:conf/icde/BorzsonyiKS01} takes a relation $r$ as input and returns the set of tuples in $r$ that are dominated by no other tuples in $r$. Dominance is defined as follows.

\begin{definition}\label{def:skyline}
Let $t$ and $s$ be tuples over a schema $S$; $t$ \emph{dominates} $s$, denoted $t\dominates s$, if, for every attribute $A\in S$, $t[A]\leq s[A]$ holds and there exists an attribute $A' \in S$ such that $t[A']<s[A']$ holds.
The \emph{skyline} $\sky(r)$ of a relation $r$ over $S$ is the set $\{t \in r \mid \nexists s\in r \logSep s\dominates t\}$.
\end{definition}

Equivalently, the skyline can be identified as the set of tuples that are top-$1$ results for at least one monotone scoring function.
\begin{equation}\label{eq:skyline-alt}
\sky(r) = \{t \mid t\in r \land \exists f\in\monotoneFunctions
\logSep \forall s\in r\logSep s\neq t \implies f(s) > f(t)\},
\end{equation}
where $\monotoneFunctions$ indicates the set of all monotone functions.

Inspired by the similarity between the expanded definition of top-$1$ set given in Equation~\eqref{eq:top-one} and the alternative definition of skyline in Equation~\eqref{eq:skyline-alt}, \cite{DBLP:journals/pvldb/CiacciaM17} introduced the notion of \emph{flexible skyline}, observing that we have the set of all monotone scoring functions in the latter and just one function in the former. The generalization of both equations to the case where we have any set of monotone scoring functions gives rise to the so-called \emph{non-dominated flexible skyline}:

\begin{definition}\label{def:flexible-skyline}
Let $t$ and $s$ be tuples over a schema $S$ and $\F$ a set of monotone scoring functions over $S$; $t$ \emph{$\F$-dominates} $s$, denoted $t\fdominates s$, iff, $\forall f\in\F\logSep f(t)\leq f(s)$ and $\exists f\in\F\logSep f(t)< f(s)$.
The \emph{non-dominated flexible skyline} $\nd(r;\F)$ of a relation $r$ with respect to $\F$ is the set $\{t \in r \mid \nexists s\in r \logSep s\fdominates t\}$.
\end{definition}

Generally, $\nd(r; \F)\subseteq \sky(r)$, since $\F\subseteq \monotoneFunctions$.
Moreover, $\nd$ generalizes both skylines and top-$1$ queries, since $\sky(r)=\nd(r;\monotoneFunctions)$ and $\TOP_1(r;f)=\nd(r;f)$ (provided, again, that $f$ causes no ties, which are not discarded by $\nd$).
With this, \cite{DBLP:conf/cikm/CiacciaM18} introduced a further generalization of the $\nd$ operator to the general case $k \geq 1$.

\begin{definition}\label{def:nd-k}
Given a relation $r$ over $S$ and a set of monotone scoring functions $\F$ over $S$, the $\nd_k(r;\F)$ operator returns the set of tuples in $r$ that are $\F$-dominated by less than $k$ tuples:
\begin{equation}\label{eq:nd-k}
\nd_k (r; \F) = \{t \in r| \nexists t_1, \ldots, t_k \in r\logSep \mathtt{diff}(t_1,\ldots,t_k)\land t_1 \fdominates t \land \ldots \land t_k \fdominates t,
\end{equation}
where $\mathtt{diff}(t_1,\ldots,t_k)$ stands for $\forall i,j\logSep (1\leq i,j\leq k \land i\neq j)\rightarrow t_i\neq t_j$.
\end{definition}

Another common notion that extends skylines to the set of tuples dominated by less than $k$ other tuples is the so-called \emph{$k$-skyband}~\cite{DBLP:journals/tods/PapadiasTFS05}, which we indicate here as $\sky_k(r)$ and define as $\sky_k (r) = \{t \in r| \nexists t_1, \ldots, t_k \in r\logSep \mathtt{diff}(t_1,\ldots,t_k)\land t_1 \dominates t \land \ldots \land t_k \dominates t\}$.
Clearly, $\nd_k(r,\monotoneFunctions)=\sky_k(r)$ and, in general, $\nd_k(r,\F)\subseteq\sky_k(r)$.

Solutions for computing $\nd_k$ in a vertically distributed setting where both sorted and random access are available were proposed in~\cite{DBLP:conf/cikm/CiacciaM18} through an algorithm called \texttt{FSA} that generalizes two classical algorithms for computing top-$k$ queries: \texttt{FA} (Fagin's Algorithm~\cite{DBLP:conf/pods/Fagin96}) and \texttt{TA} (the Threshold Algorithm~\cite{DBLP:conf/pods/FaginLN01}).
All these algorithms were defined for the so-called ``middleware scenario'', in which relation $r$ over schema $S=\{Id, A_1, \ldots, A_\dimensions\}$ is vertically distributed across relations $r_1,\ldots,r_\dimensions$ such that, for all $i\in\{1,\ldots,\dimensions\}$, $r_i$ has a schema $\{Id,A_i\}$ and is sorted on $A_i$ in ascending order (i.e., from the best to the worst value). In other words, $Id$ is meant to be a tuple identifier that can be used for reconstructing tuples through random access (and joins), and each $r_i$ is a ranked list.

We are also going to consider this kind of distribution, but we are now focusing on settings in which random access is unavailable or too expensive to be viable.
For the ``no random access'' scenario, the relevant literature has described a general pattern, called \texttt{NRA}, working for the classical top-$k$ scenario~\cite{DBLP:conf/pods/FaginLN01}. The main idea is to proceed, in parallel on all ranked lists, from top to bottom through sorted access until we have seen enough tuples to be sure that proceeding with more accesses will not change the solution.
In particular, since a tuple may have been seen only on some but not all ranked lists, for each tuple we keep track of the worst (highest) and the best (lowest) possible scores: when at least $k$ tuples have a worst score that is better than a tuple $t$'s best score, then $t$ can be dismissed. More so, when the seen tuples with the $k$ best worst scores have better worst scores than the best scores of all other tuples (seen or unseen), then \texttt{NRA} stops. To determine that no unseen tuple can beat the current top-$k$ tuples, \texttt{NRA} watches the best possible score of the so-called \emph{threshold point} $\tau$, i.e., the (virtual) tuple with attribute values corresponding to the last seen values in every ranked list: no unseen tuple can have better values than those in $\tau$, since the lists are ranked; so, if $k$ tuples already beat the threshold point, then no unseen tuple can ever enter the final result.

This pattern, however, was defined for the classical top-$k$ scenario, with just one scoring function in mind. To the best of our knowledge, no algorithm exists that has addressed the ``no random access'' scenario for the case of the non-dominated flexible skyline and its extension $\nd_k$.

\section{Flexible NRA}
\label{sec:flexible-NRA}

The classical NRA scheme consists of the following main steps:
\begin{itemize}
	\item Access items sortedly, one at a time, on all ranked lists. This will unveil the value of a tuple on some list but possibly not in all lists, so that we may not be able to compute the overall score of a tuple. We can anyhow compute bounds expressing the best and worst possible values for such a score.
	\item Keep a buffer for all the seen tuples, and for all tuples compute a worst and a best bound on their overall score.\footnote{The original algorithm uses lower and upper bounds. For generality with respect to the adopted convention, we prefer here to talk about worst and best bound.}
	\item Also compute a threshold value for the overall score that may be attained by unseen tuples.
	\item Repeat until there are $k$ tuples in the buffer whose worst bound is no worse than the best bound of all the other seen tuples and the threshold.
	\item Return such $k$ objects.
\end{itemize}

Matters are much more complex in the case of $\nd_k$ because obtaining the bounds is more challenging and also because the output size is not limited to $k$ tuples as in a top-$k$ query.

\begin{algorithm}[t]
   \scalebox{.82}
   {
    \begin{minipage}{1.33\textwidth}
	\begin{enumerate}[topsep=0pt,itemsep=-1ex,partopsep=1ex,parsep=1ex]
	   \item[Input:] \emph{Ranked lists $r_1,\ldots,r_\dimensions$, scoring functions $\F$}
	   \item[Output:] $\nd_k(r;\F)$
		\item\label{line:init} $B:=\emptyset$ \quad\emph{// buffer of tuples}
        \item\label{line:growing-phase} $\codewhile$ lists not exhausted
        \item \quad $c$ := 0\quad\emph{// a counter of tuples $\F$-dominating the threshold point}
        \item \quad make a sorted access on $r_1,\ldots,r_\dimensions$ and insert/update extracted tuples in $B$
        
        \item \quad $\tau = \langle\ell_1,\ldots,\ell_\dimensions\rangle$ \quad\emph{// threshold point: last scores on every list }
        \item \quad $\codefor\ t\ \codein$ seen tuples
        \item \quad \quad $\codeif\ wb(t) \fdominates \tau$ \quad\emph{// worst bound of $t$ $\F$-dominates $\tau$}
        \item \quad \quad \quad \codeif\ ++$c$ = $k$ \quad\emph{// threshold $\F$-dominated by $k$ tuples}
        \item\label{line:growing-phase-end} \quad \quad \quad \quad \codebreak\ to Line~\ref{line:shrinking-phase}

        \item\label{line:shrinking-phase} $\codewhile\ \codetrue$ \quad\emph{// keep digging if at least one tuple can be $\F$-dominated by $k$ tuples}
        \item \quad remove from $B$ tuples $\F$-dominated by other $k$ tuples
        \item \quad \codefor\ $s$ \codein\ $B$ \quad\emph{// candidate non-$\F$-dominated tuples}
        \item \quad \quad $c$ := 0 \quad\emph{// a counter of non-$\F$-dominance relationships}
        \item \quad \quad \codefor\ $t$ in $B\setminus \{s\}$\quad\emph{// candidate non-$\F$-dominating tuples}
        \item \quad \quad \quad \codeif\
        $bb(t)\not\fdominates wb(s)$ \quad\emph{// best bound of $t$ does not $\F$-dominate worst bound of $s$}
        \item \quad \quad \quad \quad $c$++
        \item\label{line:keep-digging} \quad \quad \codeif\ $k \leq |B|-1 - c$ // if $k$ tuples may $\F$-dominate it, we keep deepening
        \item \quad \quad \quad make a sorted access on $r_1,\ldots,r_\dimensions$ and insert/update extracted tuples in $B$
		\item \quad \quad \quad \codecontinue\ to Line~\ref{line:shrinking-phase}
		\item\label{line:shrinking-phase-end} \quad \codebreak

		\item \codereturn\ $B$
            
	\end{enumerate}	    
	\end{minipage}
   }
	\caption{Algorithmic pattern for computing $\nd_k$.}
	\label{alg:nra-nd}
\end{algorithm}

Algorithm~\ref{alg:nra-nd} shows the pseudocode that we developed for the computation of $\nd_k$.
As recognized in the pertinent literature~\cite{DBLP:journals/tods/MamoulisYCC07}, an NRA-like algorithm goes through a ``growing phase'', during which all tuples that may contribute to the final result are collected, followed by a ``shrinking phrase'', which eliminates all tuples that are not part of the result.

Essentially, during the growing phase (Lines~\ref{line:growing-phase}--\ref{line:growing-phase-end}), tuples are accessed in parallel through sorted access (and inserted in a buffer $B$) for as long as needed. In particular, we can stop this phase when we have seen at least $k$ objects whose worst bound is $\F$-dominating the threshold point $\tau$. Since the threshold is the best bound for all unseen tuples, meeting this condition means that no unseen tuple can be part of the result.
This roughly correspond to the so-called sorted access phase of Fagin's Algorithm (\texttt{FA}); yet, unlike \texttt{FA}, here we cannot proceed with random access to complete the missing parts of the extracted tuples.

This leads us to the shrinking phase (Lines~\ref{line:shrinking-phase}--\ref{line:shrinking-phase-end}), which aims to remove all tuples that are not part of the final result.
To do this, for each tuple $t$ in the buffer $B$, we keep track of how many tuples $\F$-dominate $t$ and also of how many cannot $\F$-dominate $t$. With this, if at least $k$ tuples $\F$-dominate $t$, we remove $t$ from $B$. Instead, if a tuple $s$ is not removed and there are still at least $k$ tuples that might $\F$-dominate it, then we need to continue doing sorted access, so as to discover new missing pieces of the tuples in $B$ and update the bounds. The condition on Line~\ref{line:keep-digging} expresses precisely this: the number of surviving tuples excluding $s$, i.e., $B-1$, minus the number $c$ of those tuples that do not $\F$-dominate $s$ is greater than or equal to $k$. Clearly, tuple $t$ cannot $\F$-dominate tuple $s$ if the best possible completion of $t$ does not $\F$-dominate the worst possible completion of $s$ (by completion of a tuple $t$, we mean here the tuple whose attribute values are the same as $t$'s, when available, and then the best or worst still possible for that attribute, if not available in $t$).

At the end of the shrinking phase, we are left with tuples that cannot be $\F$-dominated by $k$ or more tuples, so we have our final result.

\begin{theorem}
Algorithm~\ref{alg:nra-nd} correctly computes $\nd_k$.
\end{theorem}
\begin{proof}
The condition on Line~\ref{line:keep-digging} guarantees that $B$ does not contain any tuple that is $\F$-dominated by at least $k$ others. Therefore $B \subseteq \nd_k(r; \F)$.
In order to show that $B$ coincides with $\nd_k(r; \F)$, we need to prove that no object in $r\setminus B$ belongs to $\nd_k(r; \F)$.
Indeed, when exiting the first $\codewhile$ loop (growing phase), there are $k$ different tuples $t_1, \ldots, t_k$ $\F$-dominating the threshold $\tau$. If $t'$ is a tuple in $r\setminus B$, $t'$ cannot exceed $\tau$'s scores on any of the ranked lists, since $t'$ has not yet been met by sorted access. Therefore $t'$ is necessarily also $\F$-dominated by $t_1, \ldots, t_k$,
thus it cannot belong to $\nd_k(r; \F)$.
\end{proof}

As for performance, in top-$k$ contexts this aspect is commonly referring to the ``depth'', i.e., the number of sorted accesses made on each ranked list as an indication of the cost incurred by an algorithm.
Let $\mathtt{depth}(A,I,i)$ indicate the depth reached on $r_i$ by algorithm $A$ before returning a solution to problem $I$.
We define $\mathtt{sumDepths}(A, I )$ as $\sum_{i=1}^\dimensions \mathtt{depth}(A, I, i)$.
It is also common to assume that only tuples seen via sorted access can be returned as part of the result (no ``wild guesses'').
The notion of instance optimality characterizes those algorithms that cannot be beaten by an arbitrarily large amount by other algorithms solving the same problem. To this end, let $\mathbf{A}$ be the class of correct algorithms for computing $\nd_k$ with no wild guesses and no random access; let $\mathbf{I}$ be the set of all instances of $\nd_k$ problems.
We say that $A$ is instance optimal over $\mathbf{A}$ and $\mathbf{I}$ for the $\mathtt{sumDepths}$ cost metric if there exist constants $c_1$ (called \emph{optimality ratio}) and $c_2$ such that $\forall A'\in A\logSep \forall I\in\mathbf{I}\logSep \mathtt{sumDepths}(A, I)\leq c_1\cdot \mathtt{sumDepths}(A', I) + c_2 $.
We can now show that Algorithm~\ref{alg:nra-nd} is instance optimal for computing $\nd_k$ with no random access.

\begin{theorem}
Let $\mathbf{A}$ be the class of correct algorithms for computing $\nd_k$ with no wild guesses and no random access; let $\mathbf{I}$ be the set of all $\nd_k(r,\F)$ problems, where $\F$ is a set of monotone scoring functions.
Then Algorithm~\ref{alg:nra-nd} is instance optimal over $\mathbf{A}$ and $\mathbf{I}$.
\end{theorem}
\begin{proof}
Consider any instance $I = \langle r, \F, k\rangle$ and any algorithm $A \in\mathbf{A}$. Let $\delta_i$, $1 \leq i \leq \dimensions$, be the depth reached by $A$ on list $r_i$ when $A$ halts on $I$.
Since $A$ does not make wild guesses, all tuples seen by $A$ have been extracted by sorted access.
The best possible tuple $t$ not seen by $A$ coincides with the threshold point $\tau_A$  when $A$ halts.
Since $A$ is correct, $t$ is $\F$-dominated by at least $k$ (seen) tuples in the result.
Now, when Algorithm~\ref{alg:nra-nd} reaches depth $\delta = max\{\delta_1,\ldots,\delta_\dimensions\}$, it will have seen all tuples seen by $A$ and the corresponding threshold point $\tau^*$ will be either coinciding with or $\F$-dominated by $\tau_A$. Therefore $\tau^*$ is also necessarily $\F$-dominated by $k$ seen tuples.
This means that Algorithm~\ref{alg:nra-nd}'s stopping condition is met and, therefore, that $\mathtt{sumDepths}(\mbox{Algorithm~\ref{alg:nra-nd}}, I)\leq \delta\cdot \dimensions$, while $\mathtt{sumDepths}(A, I)=\sum^d_{i=1} \delta_i\geq\delta$,
i.e., Algorithm~\ref{alg:nra-nd} is instance optimal with an optimality ratio of $\dimensions$.
\end{proof}

We observe that our focus here is not on efficiently checking $\F$-dominance, which is used in Algorithm~\ref{alg:nra-nd} as a black box. The problem of testing $\F$-dominance has been studied extensively in~\cite{DBLP:journals/pvldb/CiacciaM17,DBLP:conf/cikm/CiacciaM18,DBLP:journals/tods/CiacciaM20} for specific classes of scoring functions and sets thereof defined by means of linear constraints on weights. We will use these results and these constraints in our experiments.

\section{Experiments}
\label{sec:experiments}

In this section, we test our implementation of Algorithm~\ref{alg:nra-nd} on a number of scenarios.

The datasets we use for the experiments comprise both synthetic and real datasets.
For synthetic datasets, we produce $\dimensions$-dimensional datasets of varying sizes and distributions according to the configurations displayed in Table~\ref{tab:operating_parameters}.
The class $\anticorrelated$ comprises datasets with values anti-correlated across different dimensions, while $\uniform$ has uniformly distributed values.

In addition to synthetic datasets, we also ran our experiments against a real dataset, $\nba$, consisting of a 2D selection of stats for $\num{190862}$ NBA games from \url{nba.com}.

In our experiments, we mainly focus on the number of $\F$-dominance tests required to compute the result and the depth incurred by the algorithm.
Besides these objective measures, we also measure execution times on a machine sporting an 8-Core Intel Core i9 with 32GB of RAM.

\begin{table}[h]
   \centering
   \caption{Operating parameters for testing efficiency (defaults in bold).}
       \begin{tabular}{|l|l|}
           \hline
               Name                           & Tested value \\
           \hline
               Distribution                        & synthetic: \uniform, \anticorrelated ; real: \nba
                \\
               Synthetic dataset size ($\size$)                & 10K, 50K, \textbf{100K}, 500K, 1M \\
               \# of dimensions ($\dimensions$)    & \textbf{2}, 3, 4\\
                $k$         & 1, 2, 5, \textbf{10}, 20, 50, 100 \\
                Spread ($\varepsilon$)  & \textsf{none}, \textbf{1\%}, 2\%, 5\%, 10\%, 20\%, 50\%, \textsf{full}\\
                Batch size ($\step$)                        & 1, 10, \textbf{100}, 1000 \\
           \hline
       \end{tabular}
   \label{tab:operating_parameters}
\end{table}

In addition to the dataset $\size$, its dimensionality $\dimensions$, and its distribution, we also study the effect of $k$ in the $\nd_k$ operator, of the constraints used for defining the set of functions $\F$, and the granularity of the sorted accesses (which may function in batches of $\step$ accesses).

For the constraints, we adopt the so-called \emph{ratio bounds constraints}, i.e., constraints of the form:
\begin{equation}\label{eq:ratio-bounds}
 \bar{w}_i(1-\varepsilon)\leq w_i\leq\bar{w}_i(1+\varepsilon),
\end{equation}
where we set $\bar{w}_i=1/\dimensions$, applied to linear scoring functions with weights $w_1,\ldots,w_\dimensions$. The values of $\varepsilon$ used in our tests are reported in Table~\ref{tab:operating_parameters}; in particular, when $\varepsilon=0$ (\textsf{none}) the weights are univocally defined and $\nd_k$ coincides with a top-$k$ query, whereas when $\varepsilon$ can vary freely (\textsf{full}), $\nd_k$ coincides with the $k$-skyband.

\begin{figure}%
\centering
\subfloat[][{\texttt{depth}}]
{\includegraphics[width=0.33\textwidth]{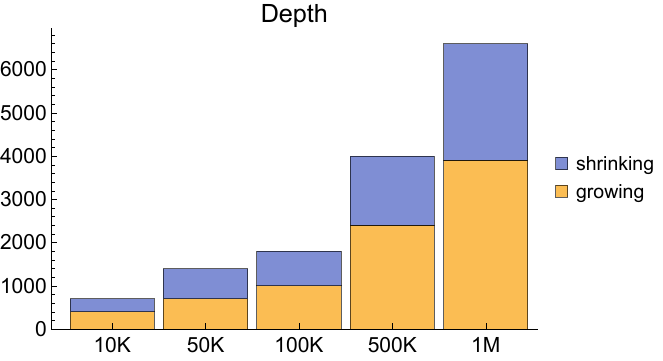}\label{fig:uni-varyingSize}}%
\subfloat[][{$\F$-dominance tests}]
{\includegraphics[width=0.33\textwidth]{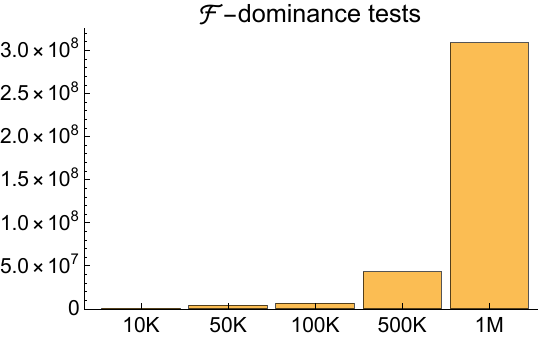}\label{fig:uni-varyingSize-fdom}}%
\subfloat[][{Time}]
{\includegraphics[width=0.33\textwidth]{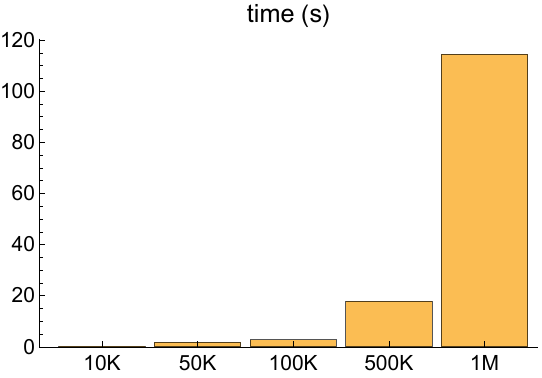}\label{fig:uni-varyingSize-time}}%
\caption{Depth (\ref{fig:uni-varyingSize}), $\F$-dominance tests (\ref{fig:uni-varyingSize-fdom}) and time ((\ref{fig:uni-varyingSize-time})) as size $\size$ varies on \uniform.}\label{fig:varyingSize}%
\end{figure}

\medskip

\medskip
\noindent\textbf{Varying the dataset size $\size$.}
The effect of the dataset size on the computation of $\nd_k$ through Algorithm~\ref{alg:nra-nd} is illustrated in Figure~\ref{fig:varyingSize}.
We varied the dataset size on \uniform using the values indicated in Table~\ref{tab:operating_parameters} and default values (indicated in bold in the table) for all other parameters. Figure~\ref{fig:uni-varyingSize} reports stacked bars for each of the tested dataset sizes, in which the lower part refers to the depth reached at the end of the growing phase, while the top part indicates the additional depth incurred during the shrinking phase.
We observe that, for datasets with uniformly distributed values, such as $\uniform$, the depth grows less than linearly with the dataset size, varying from around $7\%$ with $\size=10$K to around $0.7\%$ with $\size=1$M.
Figure~\ref{fig:uni-varyingSize-fdom} shows the number of $\F$-dominance tests that were executed in order to find the result. In this case, we can see the effect of the quadratic nature of skyline-based operators such as $\nd_k$, with the number of tests varying from $\num{387810}$ with $\size=10$K to $\num{308877008}$ with $\size=1$M.
Execution times are essentially related to the number of $\F$-dominance tests, which are the most expensive operation in the process. Figure~\ref{fig:uni-varyingSize-time} shows that such times, in seconds, vary from $0.2$s with $\size=10$K to $114.3$ with $\size=1$M, i.e., there is a growth of nearly 3 orders of magnitude, as also observed for the number of $\F$-dominance tests.

\begin{figure}%
\centering
\subfloat[][{\texttt{depth}}]
{\includegraphics[width=0.33\textwidth]{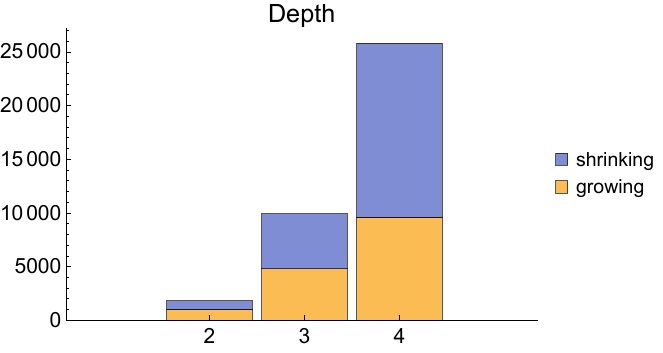}\label{fig:uni-varyingD-depth}}%
\subfloat[][{$\F$-dominance tests}]
{\includegraphics[width=0.33\textwidth]{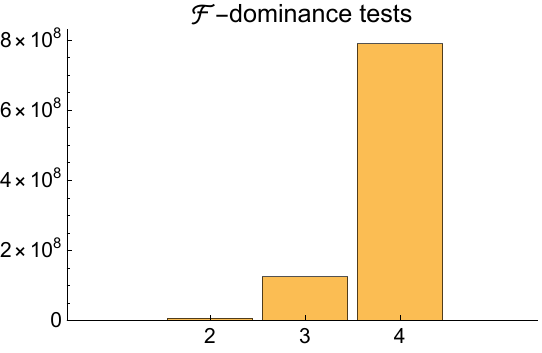}\label{fig:uni-varyingD-fdom}}%
\subfloat[][{Time}]
{\includegraphics[width=0.33\textwidth]{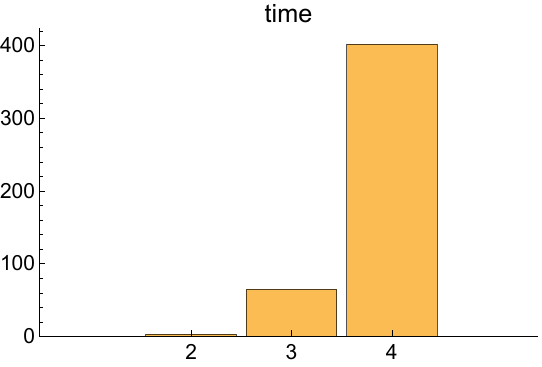}\label{fig:uni-varyingD-time}}%
\caption{Depth (\ref{fig:uni-varyingD-depth}), $\F$-dominance tests (\ref{fig:uni-varyingD-fdom}) and time (\ref{fig:uni-varyingD-time}) as $\dimensions$ varies on \uniform.}\label{fig:varyingD}%
\end{figure}

\medskip
\noindent\textbf{Varying the number of dimensions $\dimensions$.}
Figure~\ref{fig:varyingD} shows the effect of the number of dimensions $\dimensions$ on our measurements. Algorithms based on a ``no random access'' (NRA) policy heavily suffer from the so-called curse of dimensionality, since an increased number of dimensions entails less likely dominance (and $\F$-dominance) relationships, with result sets growing larger and larger. In such cases, an NRA policy essentially mandates a full scan of the dataset, since stopping criteria are met no earlier than that, thereby defeating the very purpose of ``early exit'' top-$k$ algorithms exploiting the ranking inherent in the vertically distributed sources. For these reasons, we limited our analysis to low values of $\dimensions$ (2, 3, and 4).
While the charts in Figure~\ref{fig:varyingD} are analogous to those in Figure~\ref{fig:varyingSize}, here we see that the effect of augmenting $\dimensions$ is heavier on the depth, which reaches $25\%$ with 4 dimensions, while it was just $1.8\%$ with $\dimensions =2$.
The larger number of involved tuples, with higher values of $\dimensions$, consequently entails larger numbers of $\F$-dominance tests and longer execution times, as can be seen in Figures~\ref{fig:uni-varyingD-fdom} and~\ref{fig:uni-varyingD-time}.

\begin{figure}%
\centering
\subfloat[][{\texttt{depth}}]
{\includegraphics[width=0.33\textwidth]{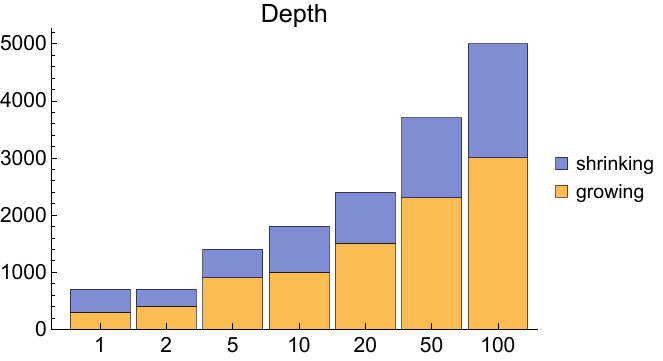}\label{fig:uni-varyingK-depth}}%
\subfloat[][{$\F$-dominance tests}]
{\includegraphics[width=0.33\textwidth]{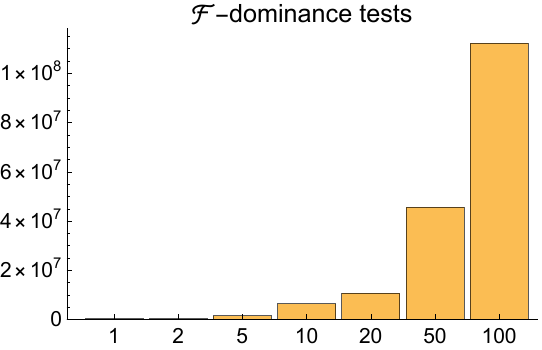}\label{fig:uni-varyingK-fdom}}%
\subfloat[][{Time}]
{\includegraphics[width=0.33\textwidth]{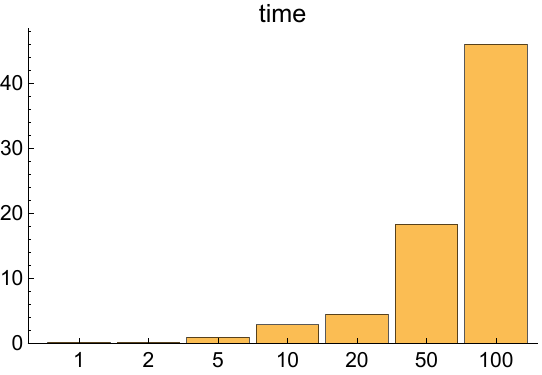}\label{fig:uni-varyingK-time}}%
\caption{Depth (\ref{fig:uni-varyingK-depth}), $\F$-dominance tests (\ref{fig:uni-varyingK-fdom}) and time (\ref{fig:uni-varyingK-time}) as $k$ varies on \uniform.}\label{fig:varyingK}%
\end{figure}

\begin{figure}%
\centering
\subfloat[][{Output size}]
{\includegraphics[width=0.33\textwidth]{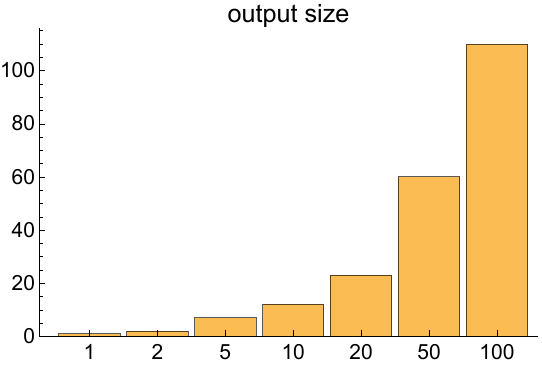}\label{fig:uni-varyingK-output}}%
\subfloat[][{Tuples after growing phase}]
{\includegraphics[width=0.33\textwidth]{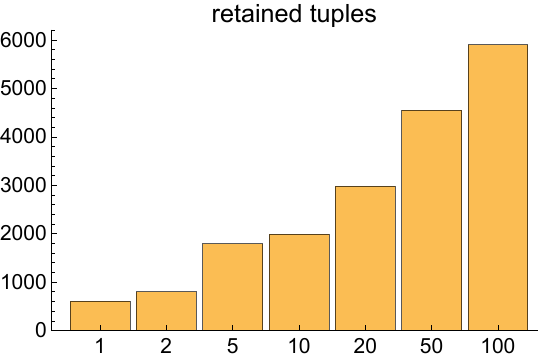}\label{fig:uni-varyingK-retained}}%
\caption{Output siez (\ref{fig:uni-varyingK-output}) and number of tuples retained after the growing phase (\ref{fig:uni-varyingK-retained}) as $k$ varies on \uniform.}\label{fig:varyingK-extra}%
\end{figure}

\medskip
\noindent\textbf{Varying $k$.}
Figure~\ref{fig:varyingK} shows the effect of $k$ on $\uniform$. While $k$ is not an exact output size in the case of $\nd_k$, it can be considered as the initial output size, which applies when $\F$ contains just one function.
We observe here that the depth grows from $0.7\%$ when $k=1$ to $5\%$ when $k=100$, i.e., less than linearly as $k$ grows.
The number of $\F$-dominance tests and the executions times are, again, tightly connected and mainly depend on the number of tuples that are retained in the growing phase and that, consequently, might need to be removed in the shrinking phase.
To this end, Figure~\ref{fig:varyingK-extra} shows how the number of retained tuples varies from right after the growing phase, i.e., when the buffer has its largest size, shown in Figure~\ref{fig:uni-varyingK-retained}, to the end of the execution, when the buffer contains the final result, whose size $|\nd_k|$ is shown in Figure~\ref{fig:uni-varyingK-output}. With our default spread value $\varepsilon=0.01$, the output size does not grow too much larger than $k$, topping $k+10$ for $k\geq 50$. Instead, the number of tuples retained at the end of the growing phase goes from just $599$ for $k=1$ to $5898$ for $k=100$, thus causing the steep increase in the number of $\F$-dominance tests shown in Figure~\ref{fig:uni-varyingK-fdom}.

\begin{figure}%
\centering
\subfloat[][{\texttt{depth}}]
{\includegraphics[width=0.33\textwidth]{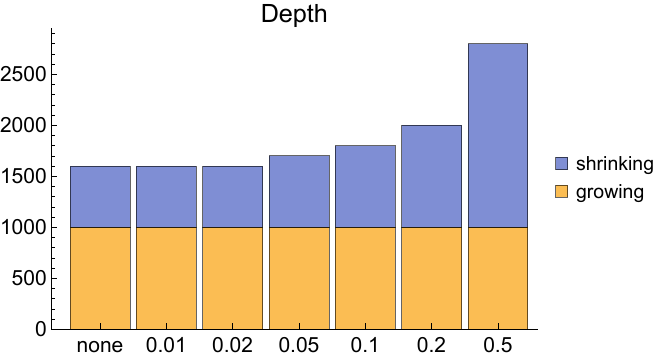}\label{fig:uni-varyingEps-depth}}%
\subfloat[][{$\F$-dominance tests}]
{\includegraphics[width=0.33\textwidth]{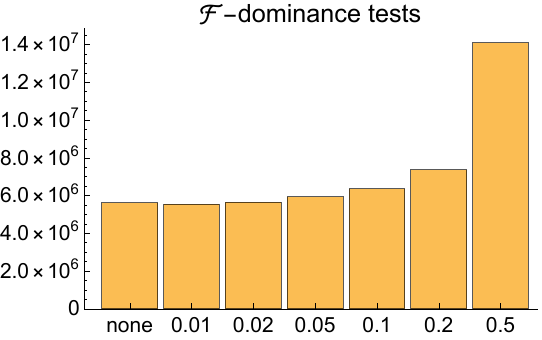}\label{fig:uni-varyingEps-fdom}}%
\subfloat[][{Time}]
{\includegraphics[width=0.33\textwidth]{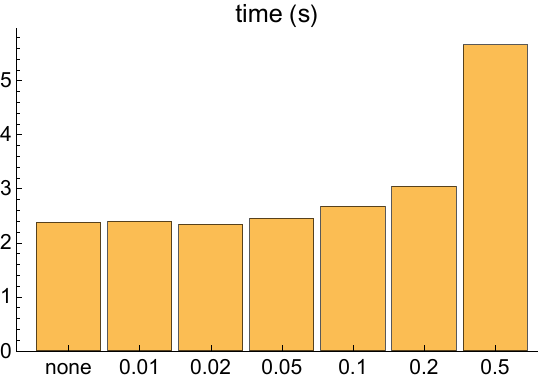}\label{fig:uni-varyingEps-time}}%
\caption{Depth (\ref{fig:uni-varyingEps-depth}), $\F$-dominance tests (\ref{fig:uni-varyingEps-fdom}) and time (\ref{fig:uni-varyingEps-time}) as spread $\varepsilon$ varies on \uniform.}\label{fig:varyingEps}%
\end{figure}

\medskip
\noindent\textbf{Varying the spread $\varepsilon$.}
The effect of $\F$, and, more precisely, of the constraints used on the weights to determine $\F$ is shown in Figure~\ref{fig:varyingEps}.
In particular, we vary the spread $\varepsilon$ of the constraints shown in~\eqref{eq:ratio-bounds} so that the $\nd_k$ operator ranges from a pure top-$k$ query (with just one linear scoring function) to a pure $k$-skyband query (with all possible linear scoring functions, which, as is well known~\cite{DBLP:journals/tods/CiacciaM20}, result-wise have the same power as all the monotone scoring functions).
Figure~\ref{fig:uni-varyingEps-depth} shows that, on the \uniform dataset with default parameter values, small values of $\varepsilon$ make $\nd_k$ deviate very little from the behavior of a top-$k$ query - and this is also confirmed in terms of $\F$-dominance tests (Figure~\ref{fig:uni-varyingEps-fdom}) and execution time (Figure~\ref{fig:uni-varyingEps-time}).
Some growth is visible starting at $\varepsilon=0.05$, and it's definitely evident for $\varepsilon=0.5$, where the depth nearly doubles with respect to the case $\varepsilon=0$ (\textsf{none}). However, we also observe that the computational toll is entirely ascribable to the shrinking phase, which requires more deepening to satisfy looser constraints.
We intentionally left out of the charts the case where $\varepsilon$ can vary freely (\textsf{full}), because there we experience an explosion in the depth (reaching $77.6\%$ vs just $2.8\%$ with $\varepsilon=0.5$) as well as in the $\F$-dominance tests (nearly $53$M vs $14$M) and execution times ($24.9$s vs $5.7$s), which would make the charts difficult to read.

\begin{figure}%
\centering
\subfloat[][{\texttt{depth}}]
{\includegraphics[width=0.33\textwidth]{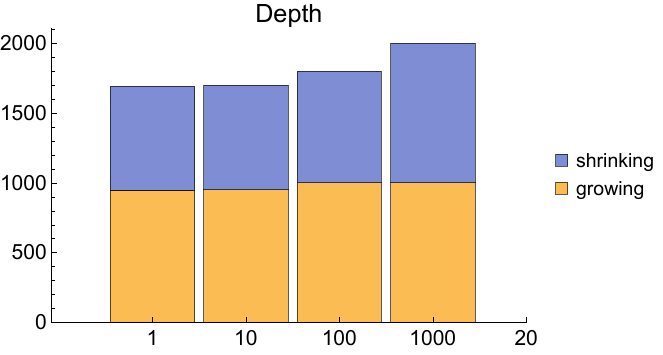}\label{fig:uni-varyingStep-depth}}%
\subfloat[][{$\F$-dominance tests}]
{\includegraphics[width=0.33\textwidth]{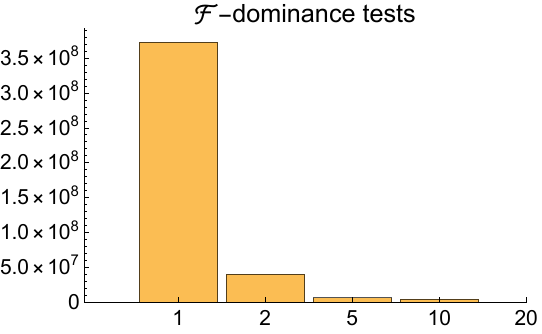}\label{fig:uni-varyingStep-fdom}}%
\subfloat[][{Time}]
{\includegraphics[width=0.33\textwidth]{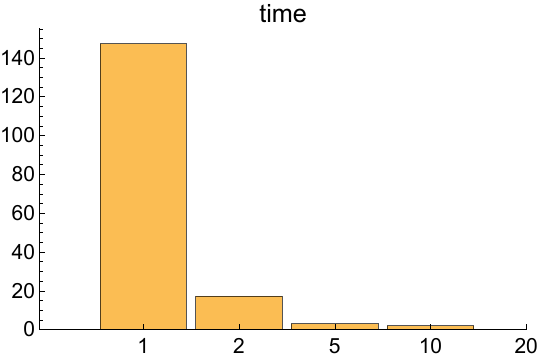}\label{fig:uni-varyingStep-time}}%
\caption{Depth (\ref{fig:uni-varyingStep-depth}), $\F$-dominance tests (\ref{fig:uni-varyingStep-fdom}) and time (\ref{fig:uni-varyingStep-time}) as step size $\step$ varies on \uniform.}\label{fig:varyingStep}%
\end{figure}

\medskip
\noindent\textbf{Varying the step $\step$.}
In order to reduce the number of times the stopping criterion is checked (which requires a high number of $\F$-dominance tests) one could try to increase the number of rows read by sorted access at once.
Normally, $\dimensions$ sorted accesses (one per list) are made and then the threshold-based stopping condition is checked (this happens during both the growing phase and the shrinking phase). Reducing the frequency of the checks to once per batch of accesses may entail a significant speed-up. Additionally, this behavior mimics the case of online services returning results in pages of a given size $\step$.
Figure~\ref{fig:varyingStep} shows the effect of varying $\step$ from the no-batch scenario $\step=1$ to $\step=1000$.
While $\step=1$ guarantees that the minimum depth will be attained during the execution, larger values will have looser guarantees on the depth, but might drastically reduce the number of incurred $\F$-dominance tests and, consequently, the execution time.
Figure~\ref{fig:uni-varyingStep-depth} shows that increasing $\step$ causes the final depth to be a multiple of the step $\step$ itself, but this negative effect may not be overall prevalent. Indeed, while with $\step=1$ we reach the minimum depth ($1691$), this only increases to $1700$ when $\step=10$ and to $1800$ when $\step=100$, while the largest increase is experienced for $\step=1000$, with a depth of $2000$. We observe that, while the depth is more or less stable for this dataset during the growing phase when $\step$ varies (with values ranging from $949$ to $1000$), larger changes are found during the shrinking phase (values from $742$ to $1000$).
However, the increase in depth is worthwhile if we look at the number of $\F$-dominance tests (Figure~\ref{fig:uni-varyingStep-fdom}) and execution times (Figure~\ref{fig:uni-varyingStep-time}): the number of $\F$-dominance tests plummets from $372$M when $\step=1$ to just $3.6$M when $\step=1000$ and times go from $147$s when $\step=1$ to just $2$s when $\step=1000$.
Due to the almost negligible difference between the cases $\step=100$ and $\step=1000$, we chose the former as the default value to use in the experiments, as it causes the lesser harm to depth.

\medskip
\noindent\textbf{Other datasets.}
As we mentioned, we also executed our experiments against the \anticorrelated family of datasets. However, due to the very nature of these datasets, an NRA-based approach like the one described in Algorithm~\ref{alg:nra-nd} is inherently ineffective. Indeed, even with the most favorable working conditions ($\dimensions=2$, $\size=10$K, $k=1$, $\varepsilon=0$, $\step=1000$), the depth explored by the algorithm is almost as large as the dataset size.
In particular, with this specific configuration, the depth was $90\%$ the dataset size, and required $58$M $\F$-dominance tests with an execution time of $27.3$s. Clearly, larger dataset sizes and less favorable conditions would determine a full scan of the dataset, with consequently higher execution times and numbers of $\F$-dominance tests.

\begin{figure}%
\centering
\subfloat[][{\texttt{depth}}]
{\includegraphics[width=0.33\textwidth]{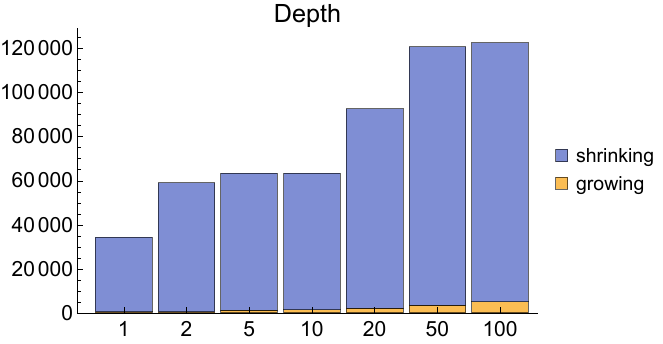}\label{fig:nba-varyingK-depth}}%
\subfloat[][{$\F$-dominance tests}]
{\includegraphics[width=0.33\textwidth]{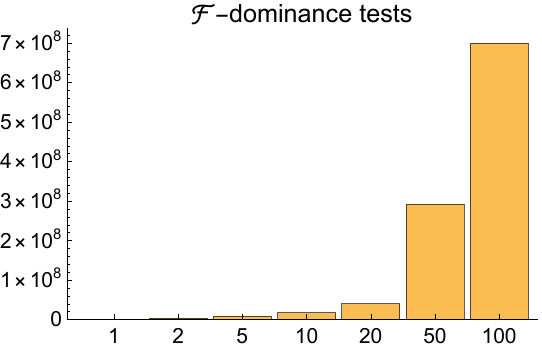}\label{fig:nba-varyingK-fdom}}%
\subfloat[][{Time}]
{\includegraphics[width=0.33\textwidth]{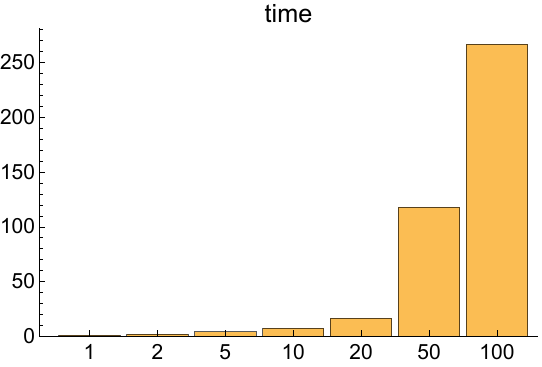}\label{fig:nba-varyingK-time}}%
\caption{Depth (\ref{fig:nba-varyingK-depth}), $\F$-dominance tests (\ref{fig:nba-varyingK-fdom}) and time (\ref{fig:nba-varyingK-time}) as $k$ varies on \nba.}\label{fig:nba-varyingK}%
\end{figure}

The real dataset $\nba$, instead, has a distribution that, albeit not uniform, is not anticorrelated either.
Figure~\ref{fig:nba-varyingK} shows the effect of $k$ on the $\nba$ dataset.
First of all, we observe that the shrinking phase is prevalent in terms of depth (Figure~\ref{fig:nba-varyingK-depth}): the total depth ranges from $\num{59000}$ to $\num{122600}$, while the depth relative to the growing phase only varies between $500$ and $5100$.
This has corresponding repercussions on the number of $\F$-dominance tests (Figure~\ref{fig:nba-varyingK-fdom}) and execution times (Figure~\ref{fig:nba-varyingK-time}), which however remain acceptable considering the dataset size and the adverse working conditions of an NRA-based algorithm.

\medskip
\noindent\textbf{Final observations.}
Our experiments show that the algorithmic scheme we proposed for computing $\nd_k$ is effective in computing the results, especially in scenarios regarding uniformly distributed data, which are more likely to allow early termination even in the absence of random access. In all other scenarios, particularly unfavorable configurations might be difficult to manage. In particular, when the growing phase first ends, the passage to the shrinking phase charges a heavy computational toll, since many tuples are in the buffer and many need to be removed, with non-negligible costs that are quadratic in the buffer size.

Our scheme proves to be versatile, in that it allows for a simple tuning of the batch size $\step$, which, while slightly worsening the total incurred depth, might heavily reduce the number of tests and, consequently, the overall execution times.

\section{Related Work}
\label{sec:related}

Three decades of work on top-$k$ queries and skyline queries have generated a large body of research, comprising numerous variants that have tried to enrich the expressivity of these queries and to overcome their main limitations.

This work lies at the culmination of a series of efforts to integrate and reconcile these two orthogonal approaches. Indeed, the $\nd_k$ operator may behave both as a top-$k$ query and as a $k$-skyband (and, ultimately, as a skyline).

Many algorithms for computing classical skylines exist in the centralized case~\cite{DBLP:conf/icde/BorzsonyiKS01,DBLP:conf/icde/ChomickiGGL03,DBLP:journals/tods/PapadiasTFS05}, but several works targeting a distributed setting are also available.
In particular, \emph{vertical partitioning}, which is also covered in the present work, was studied in~\cite{DBLP:conf/pods/Fagin98} for the so-called \emph{middleware scenario} and in several follow-up works~\cite{DBLP:conf/pods/FaginLN01,DBLP:conf/pods/SchnaitterP08}.
Many strategies also exist for \emph{horizontal partitioning}~\cite{DBLP:journals/tkde/CuiCXLSX09,DBLP:conf/edbt/MullesgaardPLZ14,ciaccia2024optimizationstrategiesparallelcomputation}.

Top-$k$ queries exist in a variety of formats~\cite{DBLP:journals/csur/IlyasBS08,DBLP:journals/pvldb/MartinenghiT10,
DBLP:journals/tods/MartinenghiT12,DBLP:journals/tkde/MartinenghiT12}, which have recently been hybridized with skylines so as to obtain the benefits of both paradigms~\cite{DBLP:journals/pvldb/CiacciaM17,DBLP:conf/sigmod/MouratidisL021}.
In particular, several previous approaches have tried to empower skylines with preferences and with the ability to limit their output~\cite{%
DBLP:journals/tods/CiacciaM20,DBLP:conf/cikm/CiacciaM18,
DBLP:conf/sigmod/MouratidisL021}.

Further limitations of top-$k$ queries reside in the expressivity of their scoring functions~\cite{DBLP:conf/sigmod/SolimanIMT11} and the robustness of their results~\cite{CM:PACMMOD2024,martinenghi2024parallelIndicators}, which the advent of flexible skylines has tried to address.

We also observe that both skylines and ranking queries are commonly included as typical parts of complex data preparation pipelines for subsequent processing based, e.g., on Machine Learning or Clustering algorithms~\cite{DBLP:conf/fqas/Masciari09,DBLP:journals/isci/MasciariMZ14,DBLP:conf/ideas/FazzingaFMF09,DBLP:journals/tods/FazzingaFFM13} as well as crowdsourcing applications~\cite{DBLP:conf/socialcom/GalliFMTN12,DBLP:conf/www/BozzonCCFMT12}.

Finally, we point out that the access constraints imposed by sorted accesses or random accesses (when available) are akin to the access limitations that have been studied extensively in the field of Web data access~\cite{%
DBLP:conf/er/CaliM08,DBLP:conf/icde/CaliM08,DBLP:conf/edbt/CaliM10,DBLP:journals/jucs/CaliCM09}.

\section{Conclusion}
\label{sec:conclusion}

In this paper, we studied the problem of computing the non-$k$-dominated flexible skyline -- a complex, skyline-based operator that encompasses common characteristics of top-$k$ queries and skyline queries, and is based on a set of scoring functions $\F$, instead of just one, as in the case of top-$k$ queries, or none, as in the case of skyline queries. In particular, we studied the application scenario in which data are vertically distributed (the so-called ``middleware scenario'') in several ranked lists and random access is not available, so that data can only be accessed, from the lists, from top to bottom.
We propose an algorithm for computing the results in two phases: a growing phase, in which all candidate results are incorporated into a buffer, and a shrinking phase, in which all tuples that are not part of the final result are removed from the buffer.
Our algorithmic scheme is not only correct, but also instance-optimal within the class of algorithms that make no random access.

We conducted extensive experiments on various configurations, targeting different datasets. We observed that, in adverse conditions, an algorithm that cannot exploit random access tends to need to consume the entire dataset, with little or no chance of experiencing an early exit. In more convenient scenarios, for instance those regarding non-anticorrelated data, we obtain acceptable execution conditions and can even exploit optimization opportunities granting a good trade-off between the final depth of the execution (i.e., a measure of the I/O cost) and the number of tests that need to be performed in order to compute the results.

Future work may try to further optimize the execution by leveraging other optimization opportunities that have been adopted for scenarios in which random access was available, such as memoization-based techniques.
Another interesting line of research regards the computation, in a no-random-access scenario, of other flexible skyline variants, such as the $\po_k$ operator.

\medskip

\noindent\textbf{Acknowledgments.} The author wishes to thank Paolo Ciaccia for insightful discussion on the manuscript.

\bibliographystyle{plain}

\end{document}